\documentclass[12pt]{article}
\usepackage{amsmath}
\usepackage{amsthm}
\usepackage{amsfonts}
\usepackage{graphicx}
\usepackage{latexsym}
\usepackage{amssymb}
\usepackage{color}
\usepackage{url}

\usepackage{slashbox}
\usepackage{xfrac}

\DeclareMathAlphabet{\mathbfsl}{OT1}{ppl}{b}{it} 

\newcommand{\omegaR}{s}

\makeatletter
\DeclareRobustCommand{\nsbinom}{\genfrac[]\z@{}}
\makeatother

\newcommand{\field}[1]{\mathbb{#1}}

\newcommand{\M}{\field{M}}
\newcommand{\F}{\field{F}}

\newcommand{\dS}{\field{S}}

\newcommand{\cB}{{\cal B}}

\newcommand{\cG}{{\cal G}}

\newcommand{\linadd}{\kern1pt\mbox{\small$\boxplus$}\kern1pt}

\newtheorem{theorem}{Theorem}

\newtheorem{remark}{Remark}

\newtheorem{corollary}{Corollary}

\theoremstyle{definition}

\newtheorem{construction}{Construction}

\begin{document}

\bibliographystyle{plain}

\title{
PIR Array Codes with Optimal Virtual Server Rate
}
\author{
{\sc Simon R. Blackburn}\thanks{Department of Mathematics, Royal Holloway University of London,
Egham, Surrey TW20 0EX, United Kingdom e-mail: {\tt s.blackburn@rhul.ac.uk}.} \hspace{1cm}
{\sc Tuvi Etzion}\thanks{Department of Computer Science, Technion,
Haifa 3200003, Israel, e-mail: {\tt etzion@cs.technion.ac.il}.
Part of this research was supported by the  BSF-NSF grant 2016692.
Part of the research was performed while
the second author visited Royal Holloway University of London under EPSRC Grant EP/N022114/1.}
}

\maketitle

\begin{abstract}
There has been much recent interest in Private information Retrieval (PIR) in models
where  a database is stored across several servers using coding techniques
from distributed storage, rather than being simply replicated. In particular,
a recent breakthrough result of Fazelli, Vardy and Yaakobi introduces the notion
of a PIR code and a PIR array code, and uses this notion to produce efficient PIR protocols.

In this paper we are interested in designing PIR array codes.
We consider the case when we have $m$ servers, with each server storing
a fraction $(1/s)$ of the bits of the database; here $s$ is a fixed rational
number with $s > 1$. A PIR array code with the \emph{$k$-PIR property} enables a $k$-server PIR protocol (with $k\leq m$) to be emulated on $m$ servers, with the overall storage requirements of the protocol being reduced. The communication complexity of a PIR protocol reduces as $k$ grows, so the \emph{virtual server rate}, defined to be $k/m$, is an important parameter. We study the maximum virtual server rate of a PIR array code with the $k$-PIR property. We present upper bounds on the achievable virtual server rate, some constructions, and ideas how to obtain PIR array
codes with the highest possible virtual server rate. In particular, we present constructions that asymptotically meet our upper bounds, and the exact largest virtual server rate is obtained when
$1 < s \leq 2$.

A $k$-PIR code (and similarly a $k$-PIR array code) is also a locally repairable code
with symbol availability $k-1$. Such a code ensures $k$ parallel reads for each information
symbol. So the virtual server rate is very closely related to the symbol availability of the code when used as a locally repairable code. The results of this paper are discussed also in this context, where subspace codes also have an important role.
\end{abstract}

\vspace{0.5cm}


\vspace{0.5cm}



\newpage
\section{Introduction}

A Private Information Retrieval (PIR) protocol allows a user to retrieve
a data item from a database, in such a way that the servers storing the data
will get no information about which data item was retrieved. The problem was introduced in~\cite{CGKS98}.
The protocol to achieve this goal assumes that the servers are curious but honest,
so they don't collude. It is also assumed that the database is error-free
and synchronized all the time. For a set of $k$ servers, the goal is to design an
efficient $k$-server PIR protocol, where efficiency is measured by the total
number of bits transmitted by all parties involved; the efficiency of the best known $k$-server PIR protocols increases as $k$ increases. This model is called
\emph{information-theoretic} PIR. There is also \emph{computational}
PIR, in which the privacy is defined in terms of the inability of a server to compute
which item was retrieved in reasonable time~\cite{KuOs97}. In this paper we will be concerned
only with information-theoretic PIR.

The classical model of PIR assumes that each server stores
a copy of an $n$-bit database, so the \emph{storage overhead}, namely the ratio between the total
number of bits stored by all servers and the size of the
database, is~$k$. However, recent work combines PIR protocols
with techniques from distributed storage (where each server stores only some of the database)
to reduce the storage overhead. This approach was first considered in~\cite{SRR14},
and several papers have developed this direction
further:~\cite{AsYa17,ALS14,BlEt16,BEP16,CHY14,FVY15,FVY15a,FGHK,RaVa16,SuJa17,TaElR16,TaElR16a,TGKFHE,VRK17,ZWWG16}.
Our discussion will follow the breakthrough approach presented
by Fazeli, Vardy, and Yaakobi~\cite{FVY15,FVY15a}.  They use a suitable array code, called a \emph{$k$-PIR array code}, to enable $m$ servers (for some $m>k$) to emulate a $k$-server PIR protocol with storage overhead significantly lower than $k$. We give more details of this approach below.

It is desirable to emulate a $k$-server PIR protocol where $k$ is as large as possible given the other parameters are fixed, since the communication complexity of the best known $k$-server PIR protocols reduces as $k$ increases. We define the \emph{virtual server rate} of a $k$-PIR array code to be $k/m$, and aim of this paper is to design code which maximise this rate, and provide corresponding upper bounds on this rate.

There has been a great deal of recent work designing codes for distributed storage across $m$ servers.
Key concepts in this application are \emph{locality}~\cite{GHSY12,KPLK14,KSPRLKKV13,RKSV14},
which is useful when we wish to restoring a server after data loss by using a small number of other servers
and \emph{symbol availability}~\cite{HYUS15,RPDV14,RPDV16,SES17,SES17a,TB2014,WZL15}, which enables
data to be read in parallel using disjoint groups of servers. Fazeli et al.~\cite{FVY15,FVY15a} observed
that a $k$-PIR array code is also a code with symbol availability $k-1$. Thus the virtual server
rate of the $k$-PIR array code is closely related to the \emph{availability rate} (defined below) of the code when used in distributed storage.

We now define the key notions discussed above more precisely.

Fazeli \emph{et al.} define a $[t\times m,p]$ \emph{$k$-PIR array code} as follows. Let $x_1,x_2,\ldots ,x_p$ be a basis of a vector
space of dimension $p$ (over some finite field $\mathbb{F}$). A \emph{$[t\times m,p]$ array
code} is simply a $t\times m$ array, each entry containing a linear combination of the basis
elements $x_i$. A $[t\times m,p]$ array code satisfies the \emph{$k$-PIR property}
(or is a \emph{$[t\times m,p]$ $k$-PIR array code}) if for every $i\in\{1,2,\ldots ,p\}$
there exist $k$ pairwise disjoint subsets $S_1,S_2,\ldots,S_k$ of columns so
that for all $j\in\{1,2,\ldots ,k\}$ the element $x_i$ is contained in the linear
span of the entries of the columns $S_j$. The following example of a (binary)
$[7\times 4,12]$ $3$-PIR array code is taken from~\cite{FVY15a}:

\[
\begin{array}{|c|c|c|c|}\hline
x_1&x_2&x_3&x_1+x_2+x_3\\\hline
x_2&x_3&x_1&x_6\\\hline
x_4&x_5&x_4+x_5+x_6&x_4\\\hline
x_5&x_6&x_8&x_9\\\hline
x_7&x_7+x_8+x_9&x_9&x_7\\\hline
x_8&x_{10}&x_{11}&x_{12}\\\hline
x_{10}+x_{11}+x_{12}&x_{11}&x_{12}&x_{10}.\\\hline
\end{array}
\]
The $3$-PIR property means that for all $i\in\{1,2,\ldots,12\}$ we can
find $3$ disjoint subsets of columns whose entries span a subspace containing $x_i$.
For example, $x_5$ is in the span of the entries in the subsets
$\{1\}$, $\{2\}$ and $\{3,4\}$ of columns; $x_{11}$ is in the span of the
entries in the subsets $\{1,4\}$, $\{2\}$ and $\{3\}$ of columns.

In the example above, many of the entries in the array consist
of a single basis element; we call such entries \emph{singletons}.

Fazeli et al use a $[t\times m,p]$ $k$-PIR array code as follows. The database
is partitioned into $p$ parts $x_1,x_2,\ldots ,x_p$, each part encoded as an
element of the finite field $\mathbb{F}$. Each of a set of $m$ servers
stores $t$ linear combinations of these parts; the $j$th server stores linear
combinations corresponding to the $j$th column of the array code. We say that
the $j$th server has $t$ \emph{cells}, and stores one linear combination
in each cell. They show that the $k$-PIR property of the array
code allows the servers to emulate all known efficient $k$-server PIR protocols.
But the storage overhead is $tm/p$, and this can be significantly smaller
than $k$ if a good array code is used.
Define $\omegaR=p/t$, so $\omegaR$ can be thought of as the reciprocal of the proportion of the database stored on each server.
For small storage overhead, we would like the ratio
\begin{equation}
\label{eq:ratio_overhead}
\frac{k}{tm/p}=\omegaR\frac{k}{m}
\end{equation}
to be as large as possible. We define the \emph{virtual server rate}
of a $[t\times m,p]$ $k$-PIR array code to be $k/m$.
The virtual server rate should not be confused with the rate of the code and the PIR rate.
The \emph{rate} of the code is ratio between the logarithm to base $q$ (when the codewords
are over the finite field GF($q$)) of the number of codewords
and the logarithm to base $q$ of the number of words in the space.
The \emph{PIR rate} is equal to the number of information bits which are obtained in a PIR scheme,
when the user is downloading one bit. In applications, we would like
the virtual server rate to be as large as possible for several reasons: when $\omegaR$, which represents the amount of storage required at each server, is fixed such schemes give small storage
overhead compared to $k$ (see (\ref{eq:ratio_overhead})); we wish to use a minimal
number $m$ of servers, so $m$ should be as small as possible;
large values of $k$, compared to $m$, are desirable, as they
lead to protocols with lower communication complexity.
We will fix the number $t$ of cells in a server, and the proportion $1/\omegaR$ of the
database stored per server and we seek to maximise the virtual server rate. Hence, we define $g(\omegaR,t)$
to be the largest virtual server rate of a $[t\times m,p]$ $k$-PIR array code when $\omegaR$ and $t$
(and so $p$) are fixed. We define $g(\omegaR)=\overline{\lim}_{t\rightarrow\infty}g(\omegaR,t)$.

Most of the analysis in~\cite{FVY15,FVY15a} was restricted to the case $t=1$.
The following two results presented in~\cite{FVY15a} are the most relevant for
our discussion. The first result corresponds to the case where each server holds
a single cell, i.e. we have a PIR code (not an array code with $t>1$).

\begin{theorem}
\label{thm:t=1}
For any given positive integer $s$, $g(s,1) = (2^{s-1})/(2^s-1)$.
\end{theorem}

The second result is a consequence of the only construction of PIR array codes
given~\cite{FVY15a} which is not an immediate consequence of the constructions
for PIR codes.

\begin{theorem}
\label{thm:FVYbound}
For any integer $s \geq 3$, we have $g(s,s-1) \geq s / (2s-1)$.
\end{theorem}

The goal of this paper is first to generalize and improve the results of Theorems~\ref{thm:t=1}
and~\ref{thm:FVYbound}
and to find codes with better virtual server rates for a given $\omegaR$.
We would like to find out the behavior of $g(\omegaR,t)$ as a function of $t$.
This will be done by providing new constructions for $k$-PIR array codes
which will imply lower bounds on $g(\omegaR,t)$ for any given pair $(s,t)$,
where $s>1$ is any rational number, and $t>1$ is an integer, such that $st$ is an integer.
This will immediately imply a related bound on $g(\omegaR)$ for various values of $\omegaR$.
Contrary to the construction in~\cite{FVY15a}, the value of $\omegaR$ in
our constructions is not necessarily an integer (this possible feature was
mentioned in~\cite{FVY15a}): each rational number greater than one will be considered.
We will also provide various upper bounds on $g(\omegaR,t)$, and related upper bounds on $g(\omegaR)$.
It will be proved that some of the upper bounds on $g(\omegaR,t)$ are tight and also our main upper
bound on $g(\omegaR)$ is tight.

We now relate these results to concepts in distributed storage system (DSS) applications.
Recall that code $C$ is called \emph{locally repairable code}
or \emph{locally recoverable code} with \emph{locality}~$r$ if every symbol from
its codewords can be recovered by at most $r$ other symbols of the codeword of a
set $R$ called a \emph{recovering set}.
A code $C$ is a locally repairable code with \emph{(symbol) locality}~$r$ and \emph{(symbol) availability}~$\kappa$
if every symbol has $\kappa$ pairwise disjoint recovering sets, each one of size at most $r$.
An extra requirement from locally repairable code is that it needs to be \emph{systematic}
which is not a requirement for a PIR code. We define the \emph{availability rate} of a code of length $m$, locality $r$ and availability $\kappa$ as $\kappa/m$. This is very similar to the virtual server rate which is defined for the same code as $(\kappa +1)/m$. When we generalise to array codes, the notions of node and symbol locality and availability become distinct: the definitions of symbol locality and symbol availability generalise in a straightforward way, with recovering sets possibly depending on the symbol stored by a node rather than just the node itself. (For array codes, node locality is often called the repair degree.)

So, how good are the codes in this paper in terms of their availability? If we are interested in locality one, we are interested in the singletons that appear in our array code. Since each node can store $t$
singletons it follows that symbol availability $\kappa$ satisfies $\kappa \leq mt/p$.
All our constructions have symbol availability smaller than $mt/p$ since most servers have non-singleton cells. Hence, the codes are not optimal in
this respect. For locality two the situation is quite different. The upper bounds on virtual server rate (and so on availability rate) that we establish do not depend on locality. Our constructions below are all of locality two, so (see Corollary~\ref{cor:1<s<2}) we already have optimal availability rate when $1 < s \leq 2$ and optimal asymptotic availability rate when $s\geq 2$. So, from the perspective of this paper,
there is no point to consider codes of
locality three or more. Of course, if we take other design parameters from
DSS into account we might need to amend our constructions.

To summarise, our notation used in the remainder of the paper is given by:
\begin{enumerate}
\item $n$ - the number of bits in the database.

\item $p$ - the number of parts the database is divided into. The parts will
be denoted by $x_1 , x_2 , \ldots , x_p$.

\item $\frac{1}{\omegaR}$ - the fraction of the database stored on a server.

\item $m$ - the number of servers (i.e. the number of columns in the array).

\item $t$ - the number of cells in a server (or the number of rows in the array); so $t=p/\omegaR$.

\item $k$ - the array code allows the servers to emulate a $k$-PIR protocol.

\item $g(\omegaR,t)$ - the largest virtual server rate of a $[t\times m,p]$ $k$-PIR array code.

\item $g(\omegaR)=\overline{\lim}_{t\rightarrow\infty}g(\omegaR,t)$.
\end{enumerate}

Though a PIR array code is formally an array of vectors, we use terminology
carried over from the application we have in mind. So we refer to a column
of this array as a server, and an entry of this column as a cell.

The information in the $t$ cells of a given server spans a subspace $V$ of $\F^p$ whose dimension
is at most $t$. It is this subspace, rather than the values in individual cells of the server,
which is important for the $k$-PIR property. Changing the cells of a given server to produce a new spanning
set for $V$, or even to replace $V$ by a larger subspace containing it, cannot harm
the $k$-PIR property. So, since the $x_i$ are linearly independent, without loss
of generality we can (and do) make two assumptions in our analysis and constructions:
\begin{itemize}
\item if $x_i$ can be derived from information in certain server alone, the singleton $x_i$ is stored as the value of one of the cells of this server;
\item the data stored in any server's cells are linearly independent, i.e. the subspace spanned by the
information in the $t$ cells has dimension $t$.
\end{itemize}

Clearly, a PIR array code is characterized by the parameters, $\omegaR$, $t$, $k$, and $m$
(the integer~$n$ must be a multiple of $p$, does not otherwise have any significant effect).
In~\cite{FVY15a}, where the case $t=1$ was considered, the goal was to find the smallest $m$ for given $\omegaR$ and $k$.
They write $M(\omegaR,k)$ for this value of $m$. The main aim in~\cite{FVY15a} was
to find bounds on $M(\omegaR,k)$, and to analyse the redundancy $M(s,k)-s$ and the storage overhead $M(\omegaR,k)/\omegaR$.
When considering PIR array codes, we have an extra parameter $t$. When $\omegaR$, $t$, and $k$ and given, the goal
is to find the smallest possible value of $m$. We write $M(s,t,k)$ for this value of $m$. Clearly, $M(s,t,k) \leq M(s,k)$, but
the main target is to find the range for which $M(s,t,k) < M(s,k)$, especially when the storage overhead is low.
For this we observe that
$$
M(s,k) \geq M(s,t,k) \geq k/g(s,t)~,
$$
which underlines the importance of the function $g(s,t)$.
Our discussion answers some of the given questions, but unfortunately not for low storage overhead
(our storage overhead is much smaller than $k$ as required, but $k$ is relatively large). Hence, our results provide an indication
of the target to be achieved, and this target is left for future work.
We will fix two parameters, $t$ and $\omegaR$,
and examine the ratio $k/m$. A high virtual server rate might require both $k$ and $m$
to be large. This might give the best storage overhead for a given $k$, but the storage overhead might well not be low: for a lower storage overhead, we probably need to compromise on a lower ratio of $k/m$.

The rest of this paper is organized as follows. In Section~\ref{sec:upper_bound} we
present a simple upper bound on the value of $g(s)$. Though this bound is attained, we prove that
$g(s,t)<g(s)$ for any fixed values of $s$ and $t$. We will also state a more complex upper bound
on $g(s,t)$ for various pairs $(s,t)$, which will be shown to be attainable for $1<s\leq 2$.
In Section~\ref{sec:constructions}
we present a few constructions, all of which are asymptotically optimal (in the sense of
having the best virtual server rate as $t\rightarrow\infty$). We believe that they are also optimal for their specific parameters ($s$, $t$, and $k$).
In Subsection~\ref{sec:s=2minus} we consider the case where $1 < s \leq 2$ and
produce a construction which attains the upper bound on $g(s,t)$.
This exact value of the virtual server rate for $1<s\leq 2$ and any admissible $t$ is given in Corollary~\ref{cor:1<s<2}.
This construction is generalized and analysed, for any rational number
$s > 1$, in Subsection~\ref{sec:Asingletons}. For small $s$ and $t$ the results are summarized in Table~\ref{tab:newPIRates}.
The asymptotic value of the virtual server rate, $g(s)$, for any rational number $s>1$ is given in Theorem~\ref{thm:asymptotic_rate}.
We provide a conclusion in Section~\ref{sec:conclude}, where problems for future research are presented.

\section{Upper Bounds}
\label{sec:upper_bound}

In this section we will be concerned first with a simple general upper bound
(Theorem~\ref{thm:upper_bound}) on the virtual server rate of a $k$-PIR array code for a fixed
value of $s$ with $s>1$. This bound cannot be attained, but is asymptotically optimal
(as $t\rightarrow\infty$). This will motivate us to give a stronger upper bound
(Theorem~\ref{thm:up_1<s<2}) on the virtual server rate $g(s,t)$ of a $[t\times m, st]$ $k$-PIR array
code for various values of $t$ that can sometimes be attained.

\begin{theorem}
\label{thm:upper_bound}
For each rational number $s > 1$ we have that $g(s) \leq (s+1)/(2s)$. There is no~$t$
such that $g(s,t) = (s+1)/(2s)$.
\end{theorem}
\begin{proof}
Suppose we have a $[t\times m,p]$ $k$-PIR array code with $p/t=s$. To prove the theorem,
it is sufficient to show that $k/m<(s+1)/2s$. Recall that we are assuming,
without loss of generality, that if $x_i$ can be derived from information
on a certain server, then the singleton~$x_i$ is stored as the value of one of the cells of this server.

Let $\alpha_i$ be the number of servers which hold the singleton $x_i$ in one of their cells.
Since each server has $t$ cells, we find that $\sum_{i=1}^p\alpha_i\leq tm$,
and so the average value of the integers $\alpha_i$ is at most $tm/p=m/s$.
So there exists $u\in\{1,2,\ldots p\}$ such that $\alpha_u\leq m/s$
(and we can only have $\alpha_u= m/s$ when $\alpha_i= m/s$ for all
$i\in\{1,2,\ldots ,p\}$). Let $S^{(1)},S^{(2)},\ldots ,S^{(k)}\subseteq\{1,2,\ldots,m\}$
be disjoint sets of servers, chosen so the span of the cells in each subset
of servers contains $x_u$. Such subsets exist, by the definition of a $k$-PIR array code.
If no server in a subset $S^{(j)}$ contains the singleton $x_u$,
the subset $S^{(j)}$ must contain at least two elements
(because we are assuming, without loss of generality, that if $x_i$ can
be derived from information on a certain server, then the singleton $x_i$ is stored
as the value of one of the cells of this server.). So at most $\alpha_u$ of the
subsets $S^{(j)}$ are of cardinality $1$.
In particular, this implies that $k\leq \alpha_u+(m-\alpha_u)/2$. Hence
\begin{equation}
\label{eq:u_bound}
\frac{k}{m} \leq \frac{\alpha_u + (m-\alpha_u)/2}{m}
=\frac{1}{2} + \frac{\alpha_u}{2m} \leq \frac{1}{2} + \frac{m/s}{2m}
= \frac{1}{2} + \frac{1}{2s} = \frac{s+1}{2s}.
\end{equation}
We can only have equality in \eqref{eq:u_bound} when $\alpha_i=m/s$ for all $i\in\{1,2,\ldots,p\}$,
which implies that all cells in every server are singletons. But then
the span of a subset of servers contains $x_i$ if and only if it contains
server with a cell $x_i$, and so $k\leq\alpha_i=m/s$. But this implies
that the virtual server rate $k/m$ of the array code is at most $1/s=2/2s$. This contradicts
the assumption that the virtual server rate of the array code is $k/m=(s+1)/2s$,
since $s>1$. So $k/m<(s+1)/2s$, as required.
\end{proof}

For non-trivial schemes we require that $p>t$, so $p=t+d$ for some positive integer~$d$. Since $p=st$, we see that $s=1+d/t$. We now provide a bound on $g(s,t)$ in terms of this integer $d$:

\begin{theorem}
\label{thm:up_1<s<2}
For any integer $t \geq 2$ and any positive integer $d$,
we have
\[
g(1 + \frac{d}{t},t) \leq \frac{(2d+1)t +d^2}{(t+d)(2d+1)}=1 - \frac{d^2+d}{(t+d)(2d+1)}.
\]
\end{theorem}
\begin{proof}
Suppose we have a $[t\times m , p]$ $k$-PIR array code with $p=t+d$.
We aim to provide an upper bound on the virtual server rate $k/m$ of this code.

For each $i\in\{1,2,\ldots, p\}$, let $S_i^{(1)},\ldots ,S_{i}^{(k_i)}\subseteq\{1,2,\ldots ,m\}$
be disjoint sets of servers, chosen so that the cells in each subset of servers
span a subspace containing $x_i$. We choose these subsets so that $k_i$ is as large
as possible subject to this condition; so
$k=\min\{k_1,k_2,\ldots  ,k_p\}\leq (\sum_{i=1}^p k_i)/p=(\sum_{i=1}^p k_i)/(t+d)$.
To prove the theorem, which asks for an upper bound on~$k/m$, it suffices to show that
\[
\sum_{i=1}^p k_i\leq\frac{(2d+1)t+d^2}{2d+1}m.
\]

Without loss of generality, we may assume that when server $j$ contains
a singleton entry~$x_i$ then $\{j\}$ is one of the subsets $S_i^{(1)},\ldots ,S_{i}^{(k_i)}$.

We say that a server is \emph{singleton} if all its cells are singletons;
otherwise we say that a server is non-singleton. Let $\ell$ be the number
of singleton servers, and let $r$ be the number of non-singleton servers. So $\ell+r=m$.

For $i\in\{1,2,\ldots, p\}$, let $\ell_i$ be the number of singleton servers with
a cell equal to $x_i$, and let $r_i$ be the number of non-singleton
servers with a cell equal to $x_i$. Since every singleton server
contains $t$ distinct singleton cells, and every non-singleton
server contains at most $t-1$ singleton cells, we see that
$\sum_{i=1}^p\ell_i=t\ell$ and $\sum_{i=1}^pr_i\leq (t-1)r$.

Let $f_i$ be the number of sets in the list $S_i^{(1)},\ldots ,S_{i}^{(k_i)}$
that are of cardinality $2$ or more, but contain at least one singleton server.
None of the sets counted by $f_i$ contain a server with a cell $x_i$,
since the sets have cardinality at least $2$. Hence $f_i\leq \ell-\ell_i$.
Every set counted by $f_i$ must involve a non-singleton server,
as cells of the form $x_u$ for $u\not=i$ can never span a space
containing $x_i$. Moreover, the non-singleton servers involved
cannot contain $x_i$ as an entry, so $f_i\leq r-r_i$.

When $i\in\{1,2,\ldots, p\}$ is fixed, there are exactly $\ell_i+r_i$ sets $S_i^{(j)}$
of size $1$, and (by definition) there are $f_i$ sets $S_i^{(j)}$ that involve
singleton servers. Every remaining set of the form $S_i^{(j)}$ must involve
at least $2$ non-singleton servers, and so there are at most $(r-r_i-f_i)/2$ sets that remain. Hence
\begin{equation}
\label{general_ki_bound}
k_i\leq \ell_i+r_i+f_i+(r-r_i-f_i)/2=\ell_i+r/2+r_i/2+f_i/2.
\end{equation}
Since $f_i\leq r-r_i$, we see that~\eqref{general_ki_bound} implies
\begin{equation}
\label{smallr_ki_bound}
k_i\leq \ell_i+r.
\end{equation}
Moreover, since $f_i\leq \ell-\ell_i$, we see that~\eqref{general_ki_bound} implies
\begin{equation}
\label{bigr_ki_bound}
k_i\leq (\ell+\ell_i)/2+(r+r_i)/2=m/2+\ell_i/2+r_i/2.
\end{equation}

Our proof now splits into two cases. First suppose that $r\leq dm/(2d+1)$.
The bound~\eqref{smallr_ki_bound} implies that
\[
\sum_{i=1}^p k_i\leq \sum_{i=1}^p\ell_i + pr = t\ell+pr = t\ell+(t+d)r=tm+dr.
\]
The right hand side is maximised when $r$ is as large as possible,
in other words when $r=dm/(2d+1)$, and so
\[
\sum_{i=1}^p k_i\leq tm + d^2m/(2d+1)=\frac{t(2d+1)+d^2}{2d+1}m,
\]
as required.

Now suppose that $r\geq dm/(2d+1)$. Then~\eqref{bigr_ki_bound} implies that
\begin{align*}
\sum_{i=1}^p k_i&\leq  pm/2+\sum_{i=1}^p\ell_i/2+\sum_{i=1}^pr_i/2\\
&\leq pm/2+t\ell/2+(t-1)r/2\\
&=((p+t)m-r)/2.
\end{align*}
The right hand side is maximised when $r$ is as small as possible,
in other words when $r=dm/(2d+1)$. So, since $p=d+t$,
\[
\sum_{i=1}^p k_i\leq (p+t)m/2-(d/(2d+1))m/2=\frac{(d+2t)(2d+1)-d}{2(2d+1)}m=\frac{t(2d+1)+d^2}{2d+1}m,
\]
as required.
\end{proof}

\begin{corollary}
For any positive integers $\delta$ and $\tau$, such that $\gcd (\delta , \tau )=1$ and for every
integer $\ell \geq 1$ we have
\[
g(1 + \frac{\delta}{\tau},t) \leq \frac{\ell \delta^2 + \tau +2 \ell t}{2 \ell \delta^2 + \delta + \tau +2 \ell t},
\]
where $t = \ell \tau$.
\end{corollary}

\section{Lower Bounds on the Virtual Server Rate}
\label{sec:constructions}

In the following subsections we will present some constructions for PIR array codes.
All of the constructions will be of a similar flavour: servers will be divided into two or more types, and much of the work will be to show how the servers that do not store a part $x_i$ may be paired up so that each pair can together recover $x_i$. In order to accomplish this, we will use Hall's marriage Theorem~\cite{Hal35} on a suitably defined graph:

\begin{theorem}
\label{thm:Hall}
A finite bipartite graph $G=(V_1\cup V_2,E)$ has a perfect matching if for each subset $X$ of $V_1$, the neighbourhood of $X$ in $V_2$ has size at least $|X|$.
\end{theorem}

\begin{corollary}
\label{cor:Hall}
A finite regular bipartite graph has a perfect matching.
\end{corollary}

\subsection{Codes with Optimal Virtual Server Rates for $1 < s \leq 2$}
\label{sec:s=2minus}

The only integer value of $s$ which is not covered by the PIR array codes in~\cite{FVY15a}
is $s=2$. Non-integer values for $s$ were not considered at all.
In this subsection we present constructions for PIR array codes when $s$
is a rational number with $1<s\leq 2$.
The construction will be generalized
in Subsection~\ref{sec:Asingletons}, where $s$ is any rational number greater than $1$,
but the special case considered here deserves separate attention for three reasons:
its description is simpler than
its generalization; the constructed PIR array code attains the bound
of Theorem~\ref{thm:up_1<s<2}, while we do not have a proof of similar result
for the generalization; and finally the analysis is
slightly different and much simpler.

Note that the number $p=st$ of parts of the database must be an integer. In particular, if we write $p=t+d$ for some positive integer $d$, we have that $s=p/t=1+d/t$, and so $1 \leq d \leq t$.

\begin{construction}{($s= 1 + d/t$ and $p= t+d$ for $t >1$, $d$ a positive integer, $1 \leq d \leq t$)}.
\label{con:1<s<2}

Let $\vartheta$ be the least common multiple of $d$ and $t$. There are two types of servers.
Servers of Type A store $t$ singletons. Each possible $t$-subset of the $p$ parts occurs $\vartheta/d$ times
as the set of singleton cells of a server, so there are $\binom{p}{t}\vartheta/d$ servers of Type A.
Each server of Type~B has $t-1$ singleton cells; the remaining cell
stores the sum of the remaining $p-(t-1)=d+1$ parts. Each possible $(t-1)$-set of singletons
occurs $\vartheta/t$ times in the set of servers of Type~B, so there are $\binom{p}{t-1}\vartheta/t$ servers of Type B.
\end{construction}

\begin{theorem}
\label{thm:1<s<2}
For any given $t>1$ and $1 \leq d \leq t$,
\[
g(1+d/t,t) \geq \frac{(2d+1)t + d^2}{(t+d)(2d+1)}.
\]
\end{theorem}
\begin{proof}
The total number of servers in Construction~\ref{con:1<s<2} is
$m =\binom{t+d}{t} \vartheta/d + \binom{t+d}{d+1} \vartheta/t$.
We now calculate $k$ such that Construction~\ref{con:1<s<2} has the $k$-PIR property.
To do this, we compute for each $i$, $1 \leq i \leq p$, a collection of pairwise disjoint sets
of servers, each of which can recover the part $x_i$.

There are $\binom{t+d-1}{t-1}\vartheta/d$ servers of Type A containing $x_i$ as a singleton cell.
Let $V_1$ be the set of $\binom{t+d-1}{t}\vartheta/d$ remaining servers of Type A.
There are $\binom{t+d-1}{t-2}\vartheta/t$ servers of Type B containing $x_i$
as a singleton cell. Let $V_2$ be the set of $\binom{t+d-1}{t-1}\vartheta/t$ remaining servers of Type~B.

We define a bipartite graph $G=(V_1 \cup V_2 ,E)$ as follows. Let $v_1\in V_1$ and $v_2\in V_2$.
Let $X_1\subseteq \{x_1,x_2,\ldots,x_p\}$ be the set of $t$ singleton cells of the server $v_1$.
Let $X_2\subseteq \{x_1,x_2,\ldots,x_p\}$ be the parts involved in the non-singleton cell of
the server $v_2$. (So $X_2$ is the set of $d+1$ parts that are not singleton cells of $v_2$.)
Since each part appears in each server of Type B, either as a singleton or in the cells
which stores a sum of $d+1$ parts, it follows that that $x_i\in X_2$.
We draw an edge between $v_1$ and $v_2$ exactly when $X_2\setminus\{x_i\}\subseteq X_1$. Note that $v_1$ and $v_2$ are joined by an edge if and only if the servers $v_1$ and $v_2$ can together recover $x_i$.

The degrees of the vertices in $V_1$ are all equal; the same is true for the vertices in $V_2$. Moreover,
$|V_1|=\binom{t+d-1}{t}\vartheta/d =
\binom{t+d-1}{t-1} \vartheta/t=|V_2|$.
So $G$ is a regular graph, and hence by Corollary~\ref{cor:Hall} there exists
a perfect matching in $G$. The  edges of this perfect matching form $|V_1|$ disjoint
sets of servers, each of which can recover $x_i$.
Thus we may take ${k=\binom{t+d-1}{t-1}\vartheta/d + \binom{t+d-1}{t-2} \vartheta/t
+ \binom{t+d-1}{t} \vartheta/d = m - \binom{t+d-1}{t} \vartheta/d}$.

Finally, some simple algebraic manipulation shows us that
\[
g(1+d/t,t) \geq \frac{k}{m} = \frac{(2d+1)t + d^2}{(t+d)(2d+1)}~.\qedhere
\]
\end{proof}

Combining Theorems~\ref{thm:up_1<s<2} and~\ref{thm:1<s<2}, we find the following.

\begin{corollary}
\label{cor:1<s<2}
$~$
\begin{enumerate}
\item[\emph{(i)}] For any given $t$ and $d$, $1 \leq d \leq t$, when $s=1+d/t$ we have
\[
g(s,t) = 1 - \frac{d^2+d}{(t+d)(2d+1)}
= \frac{t}{t+d} + \frac{d^2}{(t+d)(2d+1)}=\frac{s+1+1/d}{(2+1/d)s}~.
\]
\item[\emph{(ii)}] For any rational number $1<s \leq 2$,  we have $g(s) = (s+1)/(2 s)$.
\item[\emph{(iii)}] $g(2,t) = (3t+1)/(4t+2)$.
\end{enumerate}
\end{corollary}
\begin{proof}
Theorems~\ref{thm:up_1<s<2} and Theorem~\ref{thm:1<s<2} directly establish the first
two equalities in Part~(i) of the corollary. The final equality in Part~(i) follows from the substitution
$s= 1 + d/t = (d+t)/t$:
\[
g(s,t)=\frac{1}{s} + \frac{d^2}{st(2d+1)} =\frac{1}{s} + \frac{\frac{d}{t}}{s(2+\frac{1}{d})}
= \frac{1}{s}+\frac{s-1}{s(2+\frac{1}{d})}=\frac{s+1+\frac{1}{d}}{(2+\frac{1}{d})s}~.
\]
As $t$ and $d$ tend to infinity with $d/t$ fixed, the value of $s$ does not
change but we see that $g(s,t)\rightarrow(s+1)/(2 s)$. This establishes Part~(ii) of the corollary.
The last part~(iii) can be readily verified.
\end{proof}

Corollary~\ref{cor:1<s<2} shows that the Construction~\ref{con:1<s<2}
has optimal virtual server rate. However, we note that the number of servers used
for this code is large. PIR array codes with a smaller number of servers are more interesting for applications. We now present two constructions when $d=1$ that require a smaller, and so more practical, number of servers.

\begin{construction}{($s= 1 + 1/t$, $p=t+1$, where $t$ is odd)}
\label{con:t_t+1odd}
$~$

There are two types of servers. There are $t+1$ servers of Type A, with $t$ singleton cells.
(So exactly one part is not stored in each Type A server.) There are $(t+1)/2$ servers of Type B.
The $j$th server of Type B stores the sum $x_{2j-1} + x_{2j}$ in one cell, and the
remaining $t-1$ parts (those not equal to $x_{2j-1}$ or $x_{2j}$) as singleton cells.
\end{construction}

To recover $x_i$ in Construction~\ref{con:t_t+1odd}, we see that there are
there are $t$ servers of Type A which store $x_i$ in one of their cells as a singleton,
and $(t-1)/2$ servers of Type B which store $x_i$ as a singleton. The only server of Type B that does not
store $x_i$ as a singleton stores either $x_{i-1}+x_i$ or $x_i + x_{i+1}$.
But this server can be paired with the server of Type A that does not store $x_i$:
the Type A server stores $x_{i-1}$ or $x_{i+1}$ as appropriate, so this pair of servers can together recover $x_i$.
Thus, in this case $m=t+1 + (t+1)/2=(3t+3)/2$, and $k = (3t+1)/2$. Thus, $k/m = (3t+1)/(3t+3)$.

\begin{construction}{($s= 1 + 1/t$, $p=t+1$, where $t$ is even)}
\label{con:t_t+1even}
$~$

There are two types of servers. There are $2(t+1)$ servers of Type A, each storing $t$ singletons,
with one part not stored in each server. Each part fails to be stored on exactly two servers of Type A.
There are $t+1$ servers of Type B, where the $j$th server stores $x_j+x_{j+1}$ (subscripts taken modulo $t+1$) in one cell, and the remaining $t-1$ parts as singleton cells.
\end{construction}

To reconstruct $x_i$ using the PIR array code of Construction~\ref{con:t_t+1even}, we first note that
there are $2t$ servers of Type A and $t-1$ servers of Type B which store $x_i$ as a singleton cell.
The  two servers of Type $B$ which do not store $x_i$, store either $x_{i-1}+x_i$ or $x_i + x_{i+1}$: they
can each be paired with one of the two servers of Type A that does not store $x_i$ as a singleton, so both pairs can compute $x_i$.
Hence, we can take $k=2t+(t-1)+2=3t+1$. Since $m=3t+3$, we find that $k/m = (3t+1)/(3t+3)$.

To summarise, the virtual server rates of the PIR array codes of Constructions~\ref{con:t_t+1odd}
and~\ref{con:t_t+1even} attain the upper bound of Theorem~\ref{thm:up_1<s<2} with a small number of servers.
(In fact, it can be proved that in these constructions we have the smallest possible number of servers.)
In a recent paper~\cite{ZWWG16} which is based on the ideas presented in this paper, the authors
present some constructions with smaller number of server and optimal virtual server rate, where $t > d^2 -d$.

\subsection{A General Construction}
\label{sec:Asingletons}

Construction~\ref{con:1<s<2} can be generalized in a way that will work for
any rational number $s>1$ and any integer $t$ such that $p=st$ is an integer.
This generalized construction is presented in this subsection.
For simplicity we will define and demonstrate it first for integer values of~$s$ and later
explain the modification needed for non-integer values of $s$.

\subsubsection{The construction when $s$ is an integer}

Let $s$ and $t$ be integers with $s>1$ and $t>1$, and let $p=st$. Let $\xi_1,\xi_2,\ldots,\xi_s$ be positive integers such that
\begin{align}
\label{eqn:r_multiple}
\binom{p-t}{(r-1)t+1}\xi_r&=\binom{p-t}{rt}\xi_{r+1} \hspace{1cm}\text{ for $2\leq r\leq s-1$, and}\\
\label{eqn:one_multiple}
(s-1)\xi_1 &= \binom{p-t}{t}\xi_2.
\end{align}
Note that such integers certainly exist. (Indeed, in greater generality, given $s-1$ equations of the form $\sigma_r\xi_r=\rho_r\xi_{r+1}$ where $r=1,2,\ldots ,s-1$ and where $\sigma_r$ and $\rho_r$ are positive integers,  setting $\xi_r=\prod_{j=1}^{r-1}\sigma_j \prod_{j=r}^{s-1}\rho_j$ gives a solution to \eqref{eqn:r_multiple} and \eqref{eqn:one_multiple}.) In most situations we would like the integers $\xi_r$ to be as small as possible. So if the $\xi_r$ have a non-trivial common factor $d$, we may divide all the integers $\xi_r$ by $d$ to produce another, smaller, solution.

\begin{construction}{($s$ an integer, $s>1$)}
\label{con:integer_s}

Let $\xi_1,\xi_2,\ldots,\xi_s$ be the integers chosen above. There are $s$ types of servers: types T$_1$,T$_2,\ldots ,$T$_s$. Servers of Type T$_1$ have only singleton cells. Each subset of $t$ parts occurs $\xi_1$ times as the cells of a Type T$_1$ server, so there are $\xi_1\binom{p}{t}$ servers of Type T$_1$. Servers of Type T$_r$ with $r\geq 2$ store $t-1$ singleton cells together with a cell containing a sum of $(r-1)t+1$ of the remaining parts. Each possible subset of $t-1$ parts and sum of $(r-1)t+1$ parts occurs $\xi_r$ times, so there are $\xi_r\binom{p}{t-1}\binom{p-t+1}{(r-1)t+1}$ servers of Type T$_r$.
\end{construction}

\begin{theorem}
\label{thm:integer_s_pairing}
For any $i\in\{1,2,\ldots,p\}$, the servers in Construction~\ref{con:integer_s} that do not contain $x_i$ as a singleton may be paired in such a way that each pair may recover $x_i$.
\end{theorem}
\begin{proof}
For $r\in\{1,2,\ldots,s-1\}$, let $V_1^r$ be the set of Type T$_r$ servers whose storage does not depend on $x_i$ in any way. (So for servers in $V_1^r$, the part $x_i$ does not occur as a singleton, nor as a term in any sum.) For $r\in\{1,2,\ldots,s-1\}$, let $V_2^r$ be the set of Type T$_{r+1}$ servers that do not store $x_i$ as a singleton, but do contain $x_i$ as a term in their non-singleton cell. These sets of servers are pairwise disjoint, and the union of these sets is exactly the set of all servers not containing $x_i$ as a singleton. (To check this, note that servers of Type T$_s$ involve all parts either as a singleton or as a summand in their non-singleton cell.)

To prove the theorem, we will show that for $r\in\{1,2,\ldots,s-1\}$ the servers in $V_1^r\cup V_2^r$ may be paired in such a way that each pair may recover $x_i$.

Let $r\in\{1,2,\ldots,s-1\}$ be fixed. Define a bipartite graph $G_r$ with parts $V_1^r\cup V_2^r$ and edges defined as follows. Let $v_1\in V_1^r$ and $v_2\in V_2^r$. Let $X_1\subseteq \{x_1,x_2,\ldots,x_p\}$ be the set of $rt$ parts that occur either as a singleton cell of $v_1$ or as a summand in any non-singleton cell of $v_1$. Let $X_2\subseteq \{x_1,x_2,\ldots,x_p\}$ be the $rt+1$ parts that are summands of the non-singleton cell of $v_2$. We draw an edge from $v_1$ to $v_2$ exactly when $X_1\cup\{x_i\}=X_2$.

Note that when $v_1\in V_1^r$ and $v_2\in V_2^r$ are joined by an edge, then $v_1$ and $v_2$ can together recover $x_i$. So to prove the theorem, it suffices to show that the bipartite graph $G_r$ has a perfect matching. By Corollary~\ref{cor:Hall}, to show that $G_r$ has a perfect matching it is sufficient to prove that $G_r$ is regular. Now, by symmetry, the degrees of all vertices in $V_1^r$ are equal; the same is true for the vertices in $V_2^r$. So the theorem will follow if we can show that $|V_1^r|=|V_2^r|$ for $r\in\{1,2,\ldots,s-1\}$. When $r\geq 2$,
\begin{align*}
|V_1^r|&=\xi_r\binom{p-1}{t-1}\binom{p-t}{(r-1)t+1}\\
&=\xi_{r+1}\binom{p-1}{t-1}\binom{p-t}{rt} \text{ by~\eqref{eqn:r_multiple}}\\
&=|V_2^r|.
\end{align*}
Moreover,
\begin{align*}
|V_1^1|&=\xi_1\binom{p-1}{t}\\
&=\xi_1\frac{p-1-(t-1)}{t}\binom{p-1}{t-1}\\
&=\xi_1(s-1)\binom{p-1}{t-1} \text{ (since $p=st$)}\\
&=\xi_2\binom{p-1}{t-1} \binom{p-t}{t}\text{ by \eqref{eqn:one_multiple}}\\
&=|V_2^1|.
\end{align*}
Hence $|V_1^r|=|V_2^r|$ for $r\in\{1,2,\ldots,s-1\}$, and so the theorem follows.
\end{proof}

We now turn to finding the virtual server rate of this construction. Let $b$ be the number of servers containing $x_i$ as a singleton. (Note that $b$ does not depend on $i$.) Let $c$ be half the number of the remaining servers (in other words, the number of pairs of remaining servers). So $c=(m-b)/2$. We may take $k=b+c$, and so the virtual server rate of Construction~\ref{con:integer_s} is $k/m=(b+c)/(b+2c)$. More explicitly, we have
\begin{align*}
b&=\xi_1\binom{p-1}{t-1}+\sum_{r=2}^s\xi_r\binom{p-1}{t-2}\binom{p-t+1}{(r-1)t+1}\\
&=\binom{p-1}{t-2}\left( \xi_1\frac{p-t+1}{t-1}+\sum_{r=2}^s\xi_r\binom{p-t+1}{(r-1)t+1}\right)
\end{align*}
and, using the notation of the proof of Theorem~\ref{thm:integer_s_pairing},
\begin{align*}
c&=\sum_{r=1}^{s-1}|V_2^r|\\
&=\sum_{r=1}^{s-1}\xi_{r+1}\binom{p-1}{t-1}\binom{p-t}{rt}\\
&=\binom{p-1}{t-2}\left(\sum_{r=1}^{s-1} \xi_{r+1} \frac{p-t+1}{t-1}\binom{p-t}{rt}\right).
\end{align*}

We note that the virtual server rate depends only on the ratio between $b$ and $c$,
not on the values of $b$ and $c$ individually. In particular, since all possible
solutions $\xi_r$ to~\eqref{eqn:r_multiple} and~\eqref{eqn:one_multiple} are
equal up to a scalar multiple, the virtual server rate does not depend on our choice of
solution $\xi_r$ to these equations. Moreover, we may divide both $b$ and $c$
by $\frac{1}{t-1}\binom{p-1}{t-2}$ to obtain the following expression for the virtual server rate:

\begin{theorem}
\label{thm:integer_s_rate}
Let $s$ and $t$ be integers such that $s\geq 2$ and $t\geq 2$. Let $p=st$.
Let $\xi_1,\xi_2,\ldots,\xi_s$ be integers satisfying~\eqref{eqn:r_multiple}
and~\eqref{eqn:one_multiple}. Then the virtual server rate of Construction~\ref{con:integer_s} is $(\beta+\gamma)/(\beta+2\gamma)$ where
\begin{align*}
\beta&=\xi_1(p-t+1)+\sum_{r=2}^s(t-1)\xi_r\binom{p-t+1}{(r-1)t+1}\text{ and}\\
\gamma&=(p-t+1)\sum_{r=1}^{s-1} \xi_{r+1} \binom{p-t}{rt}.
\end{align*}
\end{theorem}

Of course the formula in Theorem~\ref{thm:integer_s_rate} is not particularly simple,
but it can easily be used to calculate the virtual server rate in any specific case.
(The comment after~\eqref{eqn:r_multiple} and~\eqref{eqn:one_multiple} gives
values for the integers $\xi_r$ that can be used.)

In a few cases, we can derive a simpler formula for the virtual server rate by an explicit
(but messy) calculation. For example, when $s=3$ we can choose $\xi_1=\binom{2t-1}{t-1}$,
$\xi_2=1$ and $\xi_3=\binom{2t}{t-1}$ to find that the construction has virtual server rate $(16t^2+7t+1)/(24t^2+15t+3)$, and so

\begin{theorem}
\label{thm:s=3}
$$
g(3,t) \geq \frac{16 t^2 +7t +1}{24t^2 +15t +3}~.
$$
\end{theorem}

Similarly, when $s=4$ we may choose $\xi_1=(t+1)\binom{3t-1}{t-1}$, $\xi_2=(t+1)$,
$\xi_3=2t$ and $\xi_4=2t\binom{3t}{t-1}$, and show that the construction has
virtual server rate $(120t^3+59t^2+12t+1)/(192t^3+128t^2+36t+4)$. In particular, we have the following theorem:

\begin{theorem}
\label{thm:s=4}
$$g(4,t) \geq \frac{120t^3 + 59 t^2 + 12t +1}{192t^3 + 128t^2 + 36t +4}~.$$
\end{theorem}

To conclude this subsubsection we present Table~\ref{tab:newPIRates},
the virtual server rates obtained by Construction~\ref{con:integer_s}
for integers $2 \leq s \leq 6$ and $1 \leq t \leq 13$, which provide lower bounds on $g(s,t)$.

\begin{table}[ht]
\begin{tabular}{|c|c|c|c|c|c|}
  \hline
  \backslashbox{t}{s} & $2$ & $3$ & $4$ & $5$  & $6$   \\ \hline
  $1$ & $\sfrac{2}{3}$ & $\sfrac{4}{7}$ & $\sfrac{8}{15}$ & $\sfrac{16}{31}$ & $\sfrac{32}{63}$ \\ \hline
  $2$ & $\sfrac{7}{10}$  & $0.6124$  & $0.57486$  & $0.55549$  & $0.54417$   \\ \hline
  $3$ & $\sfrac{5}{7}$  & $0.62878$  & $0.59057$  & $0.56978$ & $0.55693$    \\ \hline
  $4$ & $\sfrac{13}{18}$  & $0.63758$  & $0.5988$  & $0.57713$  & $0.56343$  \\ \hline
  $5$ & $\sfrac{8}{11}$  & $0.64306$  & $0.60385$  & $0.58161$ & $0.56736$  \\ \hline
  $6$ & $\sfrac{19}{26}$  & $0.64681$  & $0.60728$  & $0.58462$ & $0.57$  \\ \hline
  $7$ & $\sfrac{11}{15}$  & $0.64953$  & $0.60975$  & $0.58679$ & $0.57189$  \\ \hline
  $8$ & $\sfrac{25}{34}$  & $0.6516$  & $0.61161$  & $0.58842$ & $0.57331$  \\ \hline
  $9$ & $\sfrac{14}{19}$  & $0.65322$  & $0.61307$  & $0.58969$ & $0.57441$  \\ \hline
  $10$ & $\sfrac{31}{42}$  & $0.65452$  & $0.61424$  & $0.59071$ & $0.5753$  \\ \hline
  $11$ & $\sfrac{17}{23}$  & $0.6556$  & $0.61521$  & $0.59155$ & $0.57603$  \\ \hline
  $12$ & $\sfrac{37}{50}$  & $0.6565$  & $0.61601$  & $0.59225$ & $0.57663$  \\ \hline
  $13$ & $\sfrac{20}{27}$  & $0.65726$  & $0.61669$  & $0.59284$ & $0.57715$  \\ \hline
\end{tabular}
\centering \caption{Lower bounds on $g(s,t)$ }\label{tab:newPIRates}
\end{table}

\subsubsection{The construction when $s$ is not an integer}

Let $s>2$ be a rational number that is not an integer.  Let $t$ be an integer such that $t>1$ and $p=st$ is an integer. We take the construction above, and modify the definition of the final type of server. So we proceed as follows. Let $\xi_1,\xi_2,\ldots,\xi_{\lceil s\rceil}$ be integers satisfying the following equations:
\begin{align}
\label{eqn:r_big_general}
\xi_{\lceil s\rceil -1} \binom{p-t}{(\lceil s \rceil -2)t+1}&=\xi_{\lceil s \rceil}\\
\label{eqn:r_multiple_general}
\binom{p-t}{(r-1)t+1}\xi_r&=\binom{p-t}{rt}\xi_{r+1} \hspace{1cm}\text{ for $2\leq r\leq \lceil s\rceil-1$, and}\\
\label{eqn:one_multiple_general}
(p-t)\xi_1 &= t\binom{p-t}{t}\xi_2.
\end{align}
As before, it is clear that such integers $\xi_r$ exist, and can be easily computed.

\begin{construction}{($s$ a rational number, $s$ not an integer, $s>2$)}
\label{con:general_s}
Let $\xi_1,\xi_2,\ldots,\xi_{\lceil s\rceil}$ be the integers chosen above. There are $\lceil s\rceil$ types of servers: types T$_1$, T$_2,\ldots,$T$_{\lceil s\rceil}$. Servers of Type T$_r$ where $1\leq r\leq \lceil s\rceil -1$ are defined as in Construction~\ref{con:integer_s}. So servers of Type T$_1$ have only singleton cells, and each subset of $t$ parts occurs $\xi_1$ times as the cells of a Type T$_1$ server. There are $\xi_1\binom{p}{t}$ servers of Type T$_1$. Servers of Type T$_r$ with $2\leq r\leq \lceil s\rceil-1$ store $t-1$ singleton cells together with a cell containing a sum of $(r-1)t+1$ of the remaining parts. Each possible subset of $t-1$ parts and sum of $(r-1)t+1$ parts occurs $\xi_r$ times, so there are $\xi_r\binom{p}{t-1}\binom{p-t+1}{(r-1)t+1}$ servers of Type T$_r$. The final type of servers, those of Type T$_{\lceil s\rceil}$, have a slightly different definition: they store $t-1$ singleton cells together with a cell containing the sum of the $p-(t-1)$ remaining parts. Each possible subset of $t-1$ parts occurs $\xi_{\lceil s\rceil}$ times as the singleton cells in a server of Type T$_{\lceil s\rceil}$, so there are $\xi_{\lceil s\rceil}\binom{p}{t-1}$ servers of Type~T$_{\lceil s\rceil}$.
\end{construction}

\begin{theorem}
\label{thm:general_s_pairing}
For any $i\in\{1,2,\ldots,p\}$, the servers in Construction~\ref{con:general_s} that do not contain $x_i$ as a singleton may be paired in such a way that each pair may recover $x_i$.
\end{theorem}
\begin{proof}
The proof is very similar to the proof of Theorem~\ref{thm:integer_s_pairing}. We construct $\lceil s\rceil -1$ bipartite graphs, whose vertices together form a partition of the servers not containing $x_i$ as a singleton. The graphs
 $G_1,G_2,\ldots,G_{\lceil s \rceil-2}$ are defined as in the proof of Theorem~\ref{thm:integer_s_pairing}. But we also have a bipartite graph $G_{\lceil s\rceil-1}$ whose parts consist of the set $V_1^{\lceil s\rceil-1}$ of servers of Type T$_{\lceil s\rceil -1}$ not involving $x_i$, and the set $V_2^{\lceil s\rceil-1}$ of Type T$_{\lceil s\rceil}$ that do not store
 $x_i$ as a singleton. The edges of this graph are defined so that $v_1\in V_1^{\lceil s\rceil-1}$ and $v_2\in V_2^{\lceil s\rceil-1}$ can jointly recover $x_i$ if they are joined by an edge. (More concretely, if $X_1$ is the set of parts involved in a singleton or the sum stored by $v_1$, and $X_2$ is the set of parts involved in the sum stored by $v_2$, then we join $v_1$ and $v_2$ by an edge whenever $X_2\subseteq X_1\cup\{x_i\}$.) The equality~\eqref{eqn:r_big_general} shows that both parts of the graph contain the same number of vertices, and so the graph has a perfect matching. All other parts of the proof are essentially identical to the proof of Theorem~\ref{thm:integer_s_pairing}, and so we omit them.
\end{proof}

The following theorem provides an expression for the virtual server rate of
Construction~\ref{con:general_s}. It can be proved in the same way as Thereom~\ref{thm:integer_s_rate}.

\begin{theorem}
\label{thm:general_s_rate}
Let $s>2$ be a rational number, not an integer. Let $t$ be an integer such that $p=st$ is an integer.
Let $\xi_1,\xi_2,\ldots,\xi_{\lceil s\rceil}$ be integers
satisfying~\eqref{eqn:r_big_general}, \eqref{eqn:r_multiple_general} and~\eqref{eqn:one_multiple_general}.
Then the virtual server rate of Construction~\ref{con:general_s} is $(\beta+\gamma)/(\beta+2\gamma)$ where
\begin{align*}
\beta&=\xi_1(p-t+1)+\sum_{r=2}^{\lceil s\rceil-1}(t-1)\xi_r\binom{p-t+1}{(r-1)t+1}+(t-1)\xi_{\lceil s\rceil}\text{ and}\\
\gamma&=(p-t+1)\left(\sum_{r=1}^{\lceil s\rceil -2} \xi_{r+1} \binom{p-t}{rt}+\xi_{\lceil s\rceil}\right).
\end{align*}
\end{theorem}

\begin{remark}
We should like to remark that for most rational numbers greater than 2 we can provide constructions with the
same virtual server rate having fewer servers. We did not present these constructions as they are messy, the technique
is essentially the same, the improvement
on the number of servers is not dramatic, and we cannot prove the minimality of the number of servers in these constructions.
\end{remark}

\subsection{The Asymptotic Virtual Server Rate}

We believe that Constructions~\ref{con:integer_s} and~\ref{con:general_s} have optimal
virtual server rates for their parameters $s$ and $t$. We have not proved this,
but the theorem below shows that their virtual server rates are asymptotically optimal
as $t\rightarrow\infty$. The proof of the theorem uses a nice symmetry argument, which is an idea from~\cite{ZWWG16}.

\begin{theorem}
\label{thm:asymptotic_rate}
Let $s$ be a rational number such that $s>1$. Then $g(s)=(s+1)/2s$.
\end{theorem}
\begin{proof}
We may assume that $s>2$, by Corollary~\ref{cor:1<s<2}(ii).

The upper bound on $g(s)$ follows by Theorem~\ref{thm:upper_bound}. To show the lower bound,
let $t$ be such that $p=st$ is an integer. We note that in Construction~\ref{con:integer_s} and~\ref{con:general_s} every server stores at least $t-1$ singleton cells, and so
there are at least $(t-1)m$ singleton cells amongst the cells of the $m$ servers. Since both
Construction~\ref{con:integer_s} and~\ref{con:general_s} are symmetrical in $i\in\{1,2,\ldots,p\}$, this means that
there are at least $(t-1)m/p$ cells storing $x_i$ as a singleton. Theorem~\ref{thm:integer_s_pairing} or
Theorem~\ref{thm:general_s_pairing} shows that for every
$i\in\{1,2,\ldots,p\}$ the servers not containing $x_i$ as a singleton may be paired in such a way
that each pair may recover $x_i$. Thus
\[
k\geq (t-1)m/p + \frac{1}{2} (m-(t-1)m/p)=\frac{m}{2}(1+(t-1)/(st)).
\]
Hence the virtual server rate $k/m$ of the construction satisfies
\[
k/m\geq \frac{1}{2}(1+(t-1)/(st))\longrightarrow \frac{1}{2}(1+(1/s)=(s+1)/2s
\]
as $t\rightarrow\infty$ with $s$ fixed. This provides the lower bound on $g(s)$, as required.
\end{proof}

\section{Conclusions and Problems for Future Research}
\label{sec:conclude}

We have constructed $k$-PIR array codes with good virtual server rates for all admissible pairs $(s,t)$,
where the database is divided into $st$ parts, and each server stores $t$ linear
combinations of parts in its cells. We have also proved upper bounds on
the virtual server rate of a PIR array code. These results are strong enough to determine
the best virtual server rate of a PIR array code when $1<s\leq 2$ for all admissible choices
of $t$. Moreover, the results determine the asymptotic value for the virtual server rate
for all rational values of $s$, when $t$ is allowed to tend to infinity. These results imply analogous
results for locally repairable codes with high availability.

The research on PIR array codes is far from being complete.
Two problems, in the direction taken in this paper, for future research are enumerated below.
\begin{enumerate}
\item When $s > 2$
our upper bounds on the optimal virtual server rate fall short of our lower bounds,
although the gap tends to zero as $t$ grows.
We would like to see improved upper bounds on $g(s,t)$ when $s>2$.

\item For many pairs $(s,t)$, our constructions that achieve optimal virtual server rate have $k$ and $m$ impractically large.
Hence we ask: what is the smallest $k$ for which an optimal
virtual server rate can be obtained? More generally, we would like to have more detailed bounds that
explain the tradeoff between the parameters $t$, $k$, and $m$, for a fixed value of $s$.
\end{enumerate}

As a final comment, we formulate our problem and constructions in terms of subspaces.
We can consider the linear combination of parts in a cell as a vector in $\F_2^{p}$,
with a \emph{one} in position~$i$ if and only if $x_i$ is part of the linear combination.
This is the \emph{characteristic vector} of the linear combination.
Since we assumed without loss of generality that the $t$ linear combinations in the cells of
any given server are linearly independent, it follows that the related vectors
form a basis for a $t$-dimensional subspace of $\F_2^p$. A set $\M$
of $m$ such $t$-dimensional subspaces of $\F_2^p$ form a $k$-PIR array code if
it satisfies the following requirement:
There exists a basis $\cB$ of $\F_2^p$ and a set $\dS$ which consists of $k$ subsets
of disjoint $t$-dimensional subspaces from $\M$, such that each vector of $\cB$ is contained in the
linear span derived from the subspaces of each one of the elements of $\dS$.
This formulation translates our $k$-PIR array problems into a problem in the Grassmannian $\cG_2(p,t)$.
To our knowledge there is only one work~\cite{SES17,SES17a} which considers constructions of
locally repairable codes with large availability based on subspace codes.

We have not used this framework to obtain new $k$-PIR array codes, but this new
perspective of the problem might have its own interest as Grassmannian codes
have gained lot of interest due to applications related to error-correcting
codes in random network coding~\cite{EtSi09,EtSi13,EtSt16,EtVa11,KoKs08} and in reducing the alphabet size required
for the solution of multicast networks~\cite{EtWa16,EtWa16a}. Moreover, there has
been some interesting work done on codes for distributed storage with
subspaces~\cite{GoCa13,GTC14,Hol13,HoPo13,RaEt15,RSE16,TWB12,TWB14}.
We believe that this direction of research has lot of potential.

\section*{Acknowledgement}

The authors are indebted to Eitan Yaakobi for many helpful discussions and providing
his drafts during this research.
The authors are also indebted to the two anonymous reviewers for their insightful comments.


\end{document}